\documentclass{article}

\usepackage{microtype}
\usepackage{graphicx}
\usepackage{subfigure}
\usepackage{booktabs} 

\usepackage{hyperref}
\usepackage{comment}

\usepackage[accepted]{icml2024}

\usepackage{amsmath}
\usepackage{amssymb}
\usepackage{mathtools}
\usepackage{amsthm}

\usepackage[capitalize,noabbrev]{cleveref}

\theoremstyle{plain}
\newtheorem{theorem}{Theorem}[section]

\newtheorem{lemma}[theorem]{Lemma}

\theoremstyle{definition}
\newtheorem{definition}[theorem]{Definition}

\theoremstyle{remark}

\usepackage[disable,textsize=tiny]{todonotes}

\usepackage{graphicx}
\usepackage{wrapfig}
\usepackage{soul}
\usepackage{algorithmic}
\usepackage[parfill]{parskip}
\usepackage{dsfont}

\newcommand{\R}{{\mathbb{R}}}

\newcommand{\mC}{\mathcal{C}}
\newcommand{\mA}{\mathcal{A}}

\newcommand{\norm}[1]{\|{#1}\|}

\newcommand{\eat}[1]{}

\newtheorem{claim}{Claim}
\newtheorem{problem}{Problem}

\renewcommand{\Pr}[1]{\mathrm{Pr}\left[ {#1} \right]}
\newcommand{\Ex}[1]{\mathbb{E}\left[ {#1} \right]}
\newcommand{\IND}{\mathds{1}}
\newcommand{\ind}[1]{{\IND \left\{ #1 \right\}}}
\newcommand{\EqComment}[1]{\text{\emph{(#1)}}}

\newcommand{\rbr}[1]{\left(\,#1\,\right)}
\newcommand{\abs}[1]{\left|#1\right|}

\newcommand{\cbr}[1]{\left\{\,#1\,\right\}}

\newcommand{\ceil}[1]{\left\lceil #1 \right\rceil}

\newcommand{\cut}{\texttt{cut}}

\usepackage{color-edits}
\addauthor{kb}{magenta}
\addauthor{ms}{magenta}
\addauthor{jo}{green}

\newcommand{\wdiam}{\textnormal{wdiam}}
\newcommand{\distort}{\textnormal{distort}}
\newcommand{\contr}{\textnormal{contr}}
\newcommand{\expans}{\textnormal{expans}}

\icmltitlerunning{Dynamic Metric Embedding into $\ell_p$ Space}

\begin{document}

\twocolumn[
\icmltitle{Dynamic Metric Embedding into $\ell_p$ Space}



\icmlsetsymbol{equal}{*}

\begin{icmlauthorlist}
\icmlauthor{Kiarash Banihashem}{umdcs}
\icmlauthor{MohammadTaghi Hajiaghayi}{umdcs}
\icmlauthor{Dariusz R. Kowalski}{aug}
\icmlauthor{Jan Olkowski}{umdcs}
\icmlauthor{Max Springer}{umdmath}
\end{icmlauthorlist}

\icmlaffiliation{umdcs}{Department of Computer Science, University of Maryland, College Park, USA}
\icmlaffiliation{umdmath}{Department of Mathematics, University of Maryland, College Park, USA}
\icmlaffiliation{aug}{School of Computer and Cyber Sciences, Augusta University, Georgia, USA}

\icmlcorrespondingauthor{Max Springer}{mss423@umd.edu}

\icmlkeywords{Machine Learning, ICML}

\vskip 0.3in
]



\printAffiliationsAndNotice{} 

\begin{abstract}
    We give the first non-trivial decremental dynamic embedding 
    of a weighted, undirected graph $G$ into $\ell_p$ space.
    Given a weighted graph $G$ undergoing a sequence of edge weight increases,
    the goal of this problem is to maintain a (randomized) mapping
    $\phi: (G,d) \to (X,\ell_p)$
    from the set of vertices of the graph
    to the $\ell_p$ space such that
    for every pair of vertices $u$ and $v$,
    the expected distance between $\phi(u)$ and $\phi(v)$ in the $\ell_p$ metric
    is within a small multiplicative factor, referred to as the
    \emph{distortion}, of their distance in $G$.
    Our main result is a dynamic algorithm with expected distortion $O(\log^3 n)$
    and total update time $O\left((m^{1+o(1)} \log^2 W + Q \log n)\log(nW) \right)$, where $W$ is the
    maximum weight of the edges, $Q$ is the total number of updates and $n, m$ denote the number of vertices and edges in $G$ respectively.
    This is the first result of its kind, extending the seminal
    result of Bourgain \cite{bourgain1985lipschitz} to the growing field of dynamic algorithms.
    Moreover, we demonstrate that in the fully dynamic regime, where we
    tolerate edge insertions as well as deletions, no algorithm can
    explicitly maintain an embedding into $\ell_p$ space that has a low distortion with high probability.
    \footnote{
    An earlier version of this paper 
    claimed an expceted distortion bound of $O(\log^2 n)$ and a total update time of
    $O((m^{1+ o(1)} \log^2 W + Q) \log(nW))$. The new version obtains the slightly worse bounds of 
    $O(\log^3 n)$ and
    $O((m^{1+ o(1)} \log^2 W + Q\log n) \log(nW))$ respectively.
    }
\end{abstract}

\section{Introduction}

A low distortion embedding between two metric spaces, $M = (X,d)$ and $M' =
(X', d')$, is a 
mapping $f$ such that for every pair of points $x,y \in X$ we have 
\begin{align*}
  d(x,y) \le d'(f(x),f(y)) \le C \cdot d(x,y) \ ,
\end{align*}
where $C$ is often referred to as the
\emph{distortion} of such an embedding. Low-distortion embeddings have been extensively employed to simplify graph theoretic problems prevalent in the algorithm design literature \cite{indyk2001algorithmic}. This effectiveness stems primarily from the ability to represent any graph $G$ using a metric, wherein distances correspond to the shortest paths between two nodes. However, computing numerous graph properties within such a metric is inherently challenging. Thus, by first embedding the graph into an ``easy" metric, we can facilitate simplified problem-solving, albeit with an approximation factor determined by the distortion introduced by the embedding. For example, approximation algorithms for the sparsest cut \cite{linial1995geometry}, bandwidth \cite{blum1998semi} and buy-at-bulk \cite{awerbuch1997buy} graph problems leverage embeddings into low-distortion metric spaces to obtain their near-optimal guarantees.

In the present work, we investigate fundamental embedding problems in the
\emph{dynamic} setting, where the input graph $G$ is subject to modification at each iteration by an adversary. Specifically, we address the following question:

\begin{problem}
    Is it possible to embed any graph $G$, undergoing a dynamic sequence of edge updates, into Euclidean (and more broadly the $\ell_p$-metric) space with minimal distortion of the underlying metric's pairwise distances?
\end{problem}

Unsurprisingly, the use of randomization is essential in demonstrating that
such a data structure is indeed attainable. Most notably, we build upon the
fundamental building blocks of Bourgain, Johnson and Lindenstrauss in further
demonstrating the power of randomized decompositions of a graph to efficiently
map such an input, undergoing dynamic updates, into Euclidean space for ease of computation with only polylogarithmic expected distortion (see Section~\ref{sec:prelim} for the formal definition) of the original distances between nodes of $G$. These are the first results of their kind in the dynamic input setting.

\subsection{Motivation}

\textbf{Metric Embedding.} From a mathematical perspective, embeddings of finite metric spaces into normed spaces is a natural extension on the local theory of Banach spaces \cite{matousek02}. The goal of this area of research is to devise mappings, $f$, that preserve pairwise distances up to an additive or multiplicative \emph{distortion}. In tandem to ensuring this metric is not too heavily distorted, we also seek to ensure that the resulting embedding of a point in the original space has \emph{low-dimension} (i.e. can be represented by small number of coordinates) to ensure the representation is spacially efficient. 

Within this problem framework, the classic question is that of embedding metric spaces into \emph{Hilbert space}. Considerable literature has investigated embeddings into $\ell_p$ normed spaces (see the survey \cite{abraham2006advances} for a comprehensive overview of the main results). Most crucially, the cornerstone of the field is the following theorem by Bourgain in 1985:

\begin{theorem}[\cite{bourgain1985lipschitz}]
    For every $n$-point metric space, there exists an embedding into Euclidean space with distortion $O(\log n)$.
\end{theorem}

This landmark result is foundational in the theory of embedding into finite metric spaces. Moreover, it was further shown in \cite{linial1995geometry} that Bourgain's embedding yields an embebdding into \emph{any} $\ell_p$-metric with distorition $O(\log n)$ and dimension $O(\log^2 n)$ -- demonstrating a highly efficient and precise algorithm.

We highlight that the above results are of immense value to the field of computer science in the age of big data where the construction of appropriately sized data structures is no longer efficient (or even feasible). For instance, in the field of social media analysis, processing billions of daily tweets to identify trending topics and sentiment analysis would require impractical amounts of storage and computational resources without dimension reduction techniques like topic modeling algorithms \cite{church2017word2vec,subercaze2015metric}. It is thus essential to reduce these inputs to a more approachable metric space to prevent computational bottle-necking. 
In this paper, we present the first extension on these seminal tools to the emerging domain of \emph{dynamic algorithms}. Specifically, we maintain a polylogarithmic (expected) distortion embedding into the $\ell_p$-metric through a sequence of updates to the input graph.

\textbf{Dynamic Algorithm.} A dynamic graph algorithm is a data structure that
supports edge insertions, edge deletions, and can answer queries on certain properties of the input with
respect to the original space's metrics. While trivially one can run a
static algorithm on the graph after each update and rebuild a structure equipped to answer queries, the now large body of work on
dynamic algorithms works to devise solutions with considerably faster update
and query times. In the present work, we maintain a dynamic data
structure that both reduces the dimension of the input for ease of computation and exhibits only a modest expansion of the original metric's pairwise distances in expectation.

Similar to the fundamental motivation underlying metric embeddings, the emergence of big data has intensified the need for dynamic algorithms capable of efficiently storing representations of massive input graphs, while promptly adapting to any changes that may occur on a variety of machine learning and optimization problems \cite{bhattacharya2022efficient,dutting2023fully}.
As an illustrative example, consider the problem of maintaining connectivity information in a large graph that undergoes edge insertions and deletions -- an essential problem in the field of route planning and navigation. In a static scenario, the solution can be trivially achieved by rebuilding the shortest paths between nodes using Djikstra's algorithm on every source vertex after each update to the graph. However, it is easy to see that for connectivity graphs employed in big data systems, this procedure quickly becomes intractable. Recent advancements in the field of dynamic algorithms have revealed that it is possible to maintain such connectivity information with considerably less overhead in terms of the update time to the data structure without a large loss of accuracy for the paths \cite{bernstein2009fully,roditty2004dynamic,roditty2012dynamic}.
This capacity to adapt data structures to effectively handle diverse queries is rapidly transitioning from being merely helpful to absolutely essential. Building upon this existing body of literature, we present a novel contribution by developing a dynamic embedding structure tailored to capturing the distances between nodes in a graph, specifically within the context of the simplified $\ell_p$ metric -- a highly useful computation in the field of dimension reduction for big data. Importantly, our approach guarantees a polylogarithmic expected distortion, thereby striking a balance between efficiency and accuracy.

\subsection{Our Results}

We first explore the \emph{decremental} setting, where edge weights can only increase dynamically (i.e., nodes move further apart); this is the setting under which our primary algorithmic contributions are effective. 
For the \emph{fully} dynamic setting which allows both increases and decreases in edge weights, we show a partial negative result proving that
maintaining an embedding into the $\ell_p$-metric explicitly that has low distortion with high probability is not feasible. 
Here \emph{explicitly} maintaining an embedding means that the entire embedding is updated efficiently, rather just reporting any changes to the data structure (see Section~\ref{sec:prelim} for a more precise definition of these problem regimes).

\begin{theorem}
    There is no fully dynamic algorithm that can explicitly maintain a dynamic embedding into $\ell_p$ space with high probability. 
\end{theorem}
Though computation is efficient in the target space, we demonstrate that an adversarially selected sequence of updates to the graph can force an update of the embedding for $\Omega(n)$ nodes in each step which becomes intractable to maintain. Intuitively, this result is derived from the fact that changing a large number of pairwise distances in the $\ell_p$ metric is only possible by moving a large number of points,
while making a similar change in the input graph can be done easily by, essentially, connecting and disconnecting two components.
We expand more formally on this result in Section~\ref{sec:lower_bound}.

The main idea underpinning our primary algorithmic result is a novel combination of the static
randomized decomposition of a graph (as utilized by Bourgain) with a
decremental clustering algorithm to maintain an embedding
into $\ell_p$ space that exhibits $O(\log^{3} n)$ expected distortion and can answer distance queries with polylogarithmic update time. Our
algorithmic result is stated formally as follows. \todo{Due to the space constraints, maybe its unwise to print the full theorem statements twice since they are quite long?}
\begin{theorem}\label{thm:dyn-emb}
For every graph $G$ with max edge weight $W$ and a metric $\ell_{p}$, there is a decremental dynamic algorithm that maintains an embedding, $\rho : V \rightarrow \R^{O(\log(nW))}$, for the metric induced by the dynamically updated graph $G$ into $\ell_{p}$ space of dimension $O(\log(nW))$ that has expected (over the internal randomization of the algorithm) distortion
at most $O(\log^{3}{n})$ and its running time is at most $O\left((m^{1+o(1)} \log^2 W + Q \log n)\log(nW) \right)$ with high probability\footnote{Throughout the paper, we say that an event holds with high probability (whp for short), if its probability is at least $1-n^{-a}$ for some absolute constant $a$.}, where $Q$ denotes the total number of updates. Within this running time, the algorithm explicitly outputs all changes to the embedding and can answer distance queries between pair of vertices in $O(\log(nW))$ time.
\end{theorem}

To prove the guarantees of this algorithm, we require an alternative, constructive proof of Bourgain's lemma. 
Our algorithm is different from standard approaches to the problem which

can be classified as ``Frechet embeddings." In these embeddings, each coordinates $\rho_i(v)$ takes the form of $d_{G}(v, S_i)$ where $S_i$ is a specific set. 
However, these approaches are not suitable for the dynamic setting due to limitations in analyzing their upper bound on $\norm{\rho(u) - \rho(v)}_p$ for every given $u$ and $v$. Specifically, the distances can be maintained only approximately at best, 
prohibiting us from obtaining an upper bound.

Starting from the static case, we introduce the notion of a (random) $(\beta, R, \epsilon)$-distance preserving cut. There are two main properties of a $(\beta, R, \epsilon)$-distance preserving cut. Ignoring for now the technical $\epsilon$ parameter of this notation, the parameters $\beta$ and $R$ control the following. First, we require that the probability that two vertices are in different sets is at most $\beta$ times the distance between these vertices in $G$. Intuitively, we can expect many close vertices to be on the same side of the cut. On the other hand, for every pair of vertices whose distance in $G$ is larger than $R$, we require probability at least $\frac{1}{2}$ that they are on different sides of the cut. 
The rationale behind the latter property is that such a cut will, with constant probability, properly distribute vertices that are of distance at least $R$ in $G$. We then construct $O(\log(nW))$ such cuts, where the $i$-th cut corresponds to a different choice of the distance steering parameter $R$, i.e. $R_{i} = 2^{i}$. The final embedding is made by assigning every vertex a vector of $O(\log{nW})$ coordinates, one coordinate for corresponding to each parameter choice $R_{i}$. For every cut we denote its two sides as ``left'' and ``right''. If a vertex is on the left side of the $i$-th cut, we set its $i$-th coordinate to $0$; if it is on the right side, we set the coordinate to $R_{i}$. Using both aforementioned properties of a $(\beta, R, \epsilon)$-distance preserving cut, we show that such an assignment is an embedding with $O(\log^3{n})$ stretch.

To implement this algorithm in the dynamically changing graph $G$, we prove that $(\beta, R, \epsilon)$-distance preserving cuts can be efficiently obtained from a $(\beta, \delta)$-weak decomposition of $G$, a probabilistic graph partitioning introduced by Bartal~\cite{bartal1996probabilistic}. In this decomposition, we partition vertices of $G$ into clusters such that the distance (with respect to $G$) between every pair of vertices in a cluster is at most $\delta$, but on the other hand, for every edge the probability that this edge connects two \textit{different} clusters is at most $\beta$ times its weight. To proceed to $(\beta, R, \epsilon)$-distance preserving cuts, we augment this construction by randomly assigning each cluster to one of the two sides of the cut. In the analysis, we manage to show that such simple random assignments guarantee the properties we require from a $(\beta, R, \epsilon)$-distance preserving cut. On the other hand, provided that we are able to dynamically maintain $(\beta, \delta)$-weak decomposition of $G$, it is simple to update the random assignment after each update. To deal with a $(\beta, \delta)$-weak decomposition of $G$ under dynamic updates, we lean on the result of~\cite{forster2021dynamic} who showed how to maintain such a decomposition under \textit{edge deletions}. We observe that their framework, with few technical changes, translates to our settings.

We discuss the details of the underlying \emph{static} tools used to maintain this structure in Section~\ref{sec:alg-description} and proceed to augment these procedures to maintain edge weight updates in Section~\ref{sec:dynamic}. 
Moreover, we note that
the embedding can be used to implement a dynamic distance oracle (all-pairs shortest paths), as for each two vertices in the graph, we can
estimate their distances efficiently by calculating the distance between their embeddings.
While our distance guarantees only hold in expectation, the update
time of a distance oracle based on our algorithm nearly matches the best known bounds for the APSP problem for $O(\log^2n)$ stretch~\cite{chechik2018near, forster2023bootstrapping}, which further shows the tightness of our analysis. 

\subsection{Related Work}
\label{sec:related_work}

\paragraph{Metric Embedding.} The foundational result for the algorithmic applications of metric embedding is that of Bourgain in 1985 \cite{bourgain1985lipschitz} which embeds into any $\ell_p$ metric with logarithmic distortion. When the input metric is already the $\ell_2$ metric,
the result of Johnson and Lindenstrauss \cite{jl1984} shows that its size can be reduced to $O(\log n / \varepsilon^2)$ with $(1+\varepsilon)$ distortion for $\varepsilon > 0$. 
Recent works have studied lower bounds for the minimum number of dimensions necessary for this compression; e.g., see~\cite{larsen2017optimality}.
To the best of our knowledge, these embedding results have no analogous algorithm in the dynamic setting, which we formulate in the present work.

While $\ell_p$ space is extremely useful for functional approximation and other challenging mathematical problems, there also exists a line of research on the embeddings of an input metric to a \emph{tree metric} which inherently lends itself to dynamic problems. For embedding into these tree structures, an emphasis is placed on algorithms for \emph{probabilistic tree embeddings} (PTE) where the host metric is embedded into a \emph{distribution} of trees. Concretely, given a graph $G$, the objective is to find a distribution over a set $\tau$ of trees such that distances in $G$ do not get contracted and the expected distances over the randomly sampled tree distribution do not exceed a multiplicative \emph{stretch} of $\alpha$ (stretch here can be considered interchangeable with the concept of distortion). The preliminary work on such embeddings from Bartal \cite{bartal1996probabilistic} demonstrated that by a ``ball growing'' approach, we can embed any graph with $O(\log^2 n)$ stretch with a nearly equivalent lower bound of $\Omega(\log n)$ stretch for any such embedding. This work was later improved to obtain a PTE procedure with optimal $O(\log n)$ stretch \cite{frt2004} which has applications in problems for metric labeling \cite{kleinberg2002}, buy-at-bulk network design \cite{awerbuch1997}, vehicle routing \cite{charikar1998}, and many other such contexts \cite{bartal2004,garg2000}. %
Our dynamic emebdding procedure combines this ball growing approach with a decremental clustering procedure to efficiently maintain an embedding into the $\ell_p$-metric.

\paragraph{Dynamic Embedding.} Closely related to our work is the study of dynamic embeddings into trees. %
The work of \cite{forster2019dynamic} initiates the study on the dynamic maintenance of low-stretch such \emph{spanning} trees, devising an algorithm that yields an \emph{average distortion} of $n^{o(1)}$ in expectation with $n^{1/2 + o(1)}$ update time per operation. This result was later improved to $n^{o(1)}$ average distortion and update time bounded by $n^{o(1)}$ \cite{chechik2020dynamic}. %

The restriction of these prior works to the maintenance of spannning trees is an inherently more difficult and limited problem instance. To improve upon the above bounds, \cite{forster2021dynamic} removes this restriction and designs an embedding procedure that guarantees an \emph{expected} distortion of $n^{o(1)}$ in $n^{o(1)}$ update time, or $O(\log^4 n)$ stretch with $m^{1/2 + o(1)}$ update time when embedding into a distribution of trees. This work also devises a decremental clustering procedure that we build upon in the present work to devise our embeddings. We additionally note that the expected distortion objective more closely aligns with our primary result, however our embedding into the $\ell_p$ -metric is better suited for the class of NP-hard optimization problems whose approximation algorithms rely on the geometry of Euclidean space such as sparsest cut \cite{arora2005euclidean,aumann1998log,chawla2008embeddings}, graph decompositions \cite{arora2009expander,linial1995geometry}, and the bandwidth problem \cite{dunagan2001euclidean,feige1998approximating,krauthgamer2004measured}. 

Similar to the present work is the study of \emph{dynamic distance oracles} as originally studied by \cite{thorup2005approximate} in the static setting, and later extended to the decremental setting with a data structure which maintains the distance between any two points from the input metric with $O(2k-1)$ stretch, $\tilde{O}(mn)$ total update time and $O(m + n^{1+1/k})$ space (where $k$ is any positive integer) \cite{roditty2012dynamic}. %
This result can be further improved to a distortion of $1+\varepsilon$ with $\tilde O(n^2)$ space for every $\varepsilon > 0$. \cite{chechik2018near} further present a decremental algorithm for the all pairs shortest path (APSP) problem which admits $(2+\varepsilon)k - 1$ distortion with total update time of $O(mn^{1/k + o(1)}\log(nW))$ and query time $O(\log\log(nW))$. Our embedding which generalizes this notion of distance oracle yields a nearly equivalent update time for $O(\log^2 n)$ expected distortion, further demonstrating the tightness of our analysis. 

In the next section, we precisely define the mathematical framework and formalization within which our algorithmic techniques reside. 

\section{Model and Preliminaries} \label{sec:prelim}

Let $G = (V,E)$ be a weighted, undirected graph on $n$ vertices 
with (at most) $m$ edges of positive integer weights in the range from
1 to $W$, where $W$ is a fixed parameter known to the algorithm.
For an edge $(u,v) \in E$, we denote
its weight by $w_G(u,v)$. 
For every pair of nodes $u,v\in V$, let $d_G(u,v)$ be
the length of the shortest weighted path between nodes $u,v$
in $G$, where we define the weight of a path as the sum of the weights
of its edges.
Throughout, we let $\Delta$ denote a power of two that is always larger than the diameter of the graph; note that $\Delta \in O(nW)$.
We note that $(V, d_{G})$ is a metric space.

Given a set of vertices $V'
\subseteq V$, we define the \emph{weak diameter} of $V'$ as the maximum
distance 
between the vertices of $V'$ in the original graph, i.e.,
\begin{math}
    \wdiam(V') = \sup_{u, v \in V'}d_{G}(u, v)
    \ .
\end{math}
For all $u \in V$ and 
$r \ge 0$, let $B_G(u,r)$ denote the set of all vertices that
are within distance $r$ from $u$ in the graph $G$, i.e.,
$B_{G}(u, r) := \cbr{v \in V: d_{G}(u, v) \le r}$. 

\textbf{Metric Embedding.} 
The objective of this paper is to 
construct and maintain an embedding 
of the metric
defined by an input graph $G$ 
to an $\ell_p$ metric space without distorting
the original distances by too much.
More formally, given a metric space $(X, d_{X})$,
an injective mapping $f : G \rightarrow X$ is called an
\emph{embedding}, from $G$ into $X$. 
We define the \emph{expansion} (or stretch) and the \emph{contraction}
of the embedding $f$, respectively, as: 
\begin{align*}
  \expans(f) &= \sup_{u,v\in V; u \ne v}\frac{d_{X}(f(u), f(v))}{d_{G}(u, v)} 
\\
  \contr(f) &= \sup_{u,v\in V; u \ne v}\frac{d_{G}(u, v)}{d_{X}(f(u), f(v))}
  \ .
\end{align*}
We define the distortion of the embedding $f$ as
$\distort(f) = \expans(f) \cdot \contr(f)$.
Note that any embedding $f$ satisfies
\begin{math}
  \frac{1}{\contr(f)}\cdot d_{G}(u, v) \le d_{X}(f(u), f(v)) \le \expans(f) \cdot d_{G}(u, v).
\end{math}
The embeddings in this paper are random functions, and are constructed by randomized algorithms.
Given a random embedding $f: V \to X$,
we define its \emph{expected distortion}
as the smallest value $\alpha > 0$
for which
there exist positive values $a, b$ satisfying
$ab=\alpha$ such that for all $u, v \in V$:
\footnote{Throughout the paper, we mostly consider $a=1$. As such, we sometimes use distortion and stretch interchangeably since we are only concerned with the expansion of distances between points.}
\begin{align}
  \frac{1}{a} \cdot d_{G}(u, v) \le \Ex{d_{X}(f(u), f(v))} \le b \cdot d_{G}(u, v)
    \ .
  \label{eq:dist_expt}
\end{align}
In this paper, we focus on embeddings into the $\ell_p$ metric space.
In this metric space, the ground set $X$ equals $\R^{d}$, for some positive integer
$d$, and for every pair of points $x, y \in X$, the distance $d_{X}$ is defined as
\begin{align*}
  d_{X}(x, y) = \norm{x - y}_{p} = \rbr{\sum_{i=1}^{d} |x_i - y_i|^{p}}^{1/p},
\end{align*}
where $x_i$ and $y_i$ refer to the $i$-th coordinate of $x$ and $y$, respectively.

\textbf{Dynamic Model.} 
We consider a model where
the underlying input graph $G$ undergoes
a sequence of updates as specified by an \emph{oblivious} adversary.
We assume that the adversary knows the algorithm, but does
not have access to the random bits the algorithm uses.
We use $G_0, G_1, \dots$ to denote the corresponding sequence of graphs,
where $G_i$ refers to the graph after $i$ updates. 
Throughout, we will use $Q$ to denote the total number of updates to an input graph.
This sequence is fixed by the adversary before the execution of the algorithm, but is revealed to the algorithm gradually, one by one.
Our goal is to explicitly maintain
an embedding after each update, as formally defined below:
\begin{definition}[Maintain]
  We say that a dynamic algorithm $\mA$ \emph{explicitly maintains} an embedding
  of the input graph into a metric space $(X, d_{X})$
  if there exists a sequence of mappings
  $\phi_0, \phi_1, \dots$ where
  $\phi_i: V \to X$ and
  $\mA$ outputs the changes in $\phi$ after every update.
  Formally, after the update $t$,
  the algorithm should output $v$ and $\phi_t(v)$ for all $v$ such that $\phi_t(v) \ne \phi_{t-1}(v)$.
\end{definition}

We operate in the \emph{decremental} setting and assume that each update takes the form
of an edge weight increase, i.e., for an edge $(u, v)\in E$,
the value of $w_{G}(u, v)$ increases.
We note that this is slightly different from the standard
definition of the decremental setting which permits the \emph{deletion} of edges in the input graph. The deletion of an edge can lead the input graph to potentially become disconnected, which means we may have $d_{G_{t}}(u, v) =\infty$ for some time step $t$ and $u, v\in V$.
This is problematic, however, because regardless
of the value of $\phi_t(u)$ and $\phi_t(v)$, we will always have
$\norm{\phi_t(u) - \phi_t(v)}_{p} < \infty$ because the $\ell_p$
metrics do not allow for infinite distances.
This in turn means that we cannot satisfy
the bounds for expected distortion (Equation \eqref{eq:dist_expt}),
and as such cannot design a low-distortion embedding.
To avoid
this issue, we restrict the updates to edge weight increases only, and we note that in practice the removal an edge can be simulated by choosing a large
$W$ as the dependence of our bounds on $W$ will be polylogarithmic. Thus, edge weight increases serve as a necessary stand-in for edge deletions as both will lead to pairwise distances increasing. %

In the section that follows, we will show that maintaining a fully dynamic embedding, where edge weights are subject to both increases and decreases, that
has low distortion with high probability is unfeasible in the $\ell_p$-metric space if the distortion bounds hold. 
This limitation underpins the rationale for the above decremental problem setting we introduce.

\section{Lower Bound for Explicit Maintenance of Fully Dynamic Embeddings} \label{sec:lower_bound} %

We first present an (oblivious) adversarial construction of edge weight modifications to a graph in the \emph{fully dynamic} model that cannot be \emph{explicitly} maintained in the geometry of $\ell_p$ space without needing to modify the embedding for every node in the original graph. We highlight that this is a high probability result whereas the main algorithmic results we obtain hold in expectation.

\begin{theorem}
    Any fully dynamic algorithm that 
    maintains an embedding into the $\ell_p$-metric space which guarantees a distortion of at most $o(W)$ with high probability must have an update time at least $\Omega(n)$.
\end{theorem}
\begin{proof}
    Let $\mA$ be a fully dynamic algorithm which guarantees a stretch of at most $o(W)$ with high probability.
    Consider an input graph $G$ that consists
    of two separate complete graphs on $n$ vertices, 
    $H$ and $H'$, comprised of unit edges.
    Further consider two fixed vertices 
    $v \in H, v'\in H'$.
    If there
    is a unit edge between these two vertices, then the distance
    of all elements in 
    $H$ and $H'$
    is at most 3 in the graph metric, and therefore in the $\ell_p$ embedding cannot be more than $o(W)$.

    Now, assume an adversary increases the edge weight connecting the vertices $v$ and $v'$ to a value of $W$. In the original graph metric, all pairwise distances between the nodes of $H$ and $H'$ must now be at least $W$. Therefore, the embedded points of one cluster ($H$ or $H'$) must be updated so as to not contract the original metric and maintain the distortion of at most $W$ with high probability (see Figure~\ref{fig:lower_bound} for a depiction of this construction). Therefore, the algorithm $\mA$ must update the embedding for all $n$ nodes of one of the complete components of $G$ to be at least $W$ away from the other with respect to the $\ell_p$ norm and satisfy the distortion constraints with high probability.
    Thus, we charge at least $\Omega(n)$ to the update time of $\mathcal A$ in the worst case for the maintenance of the embedding.
    \footnote{Formally, for each pair of vertices in $H \times H'$, at least one of them needs to be updated. Since there are $\Omega(n^2)$ pairs and each vertex update resolves the issue for $O(n)$ pairs, we need $\Omega(n)$ vertex updates.}
    Moreover, we cannot amortize this worst case update occurrence since, in the subsequent iteration, the adversary can change the edge weight back 1 and repeat the cycle -- resulting in $\Omega(n)$ updates per iteration.
    \eat{
    \mscomment{revising this proof...}
    If this connected edge is increased by an adversary to a value of $W$ however, the embedded distance must then proportionately increase and subsequently all nodes in the two connected components embeddings must be updated. 

    Therefore, if you consider any pair of vertices 
    $u \in H, u' \in H'$, 
    then if we keep adding and removing the edge between $v_1, v_2$, then at least one of $u_1$ and $u_2$ should move each time. It follows that at least $O(n)$ changes are necessary each time there is an update. Note that the changes can't be amortized because we can keep inserting and deleting.
    }
\end{proof}
\begin{figure}[b]
\includegraphics[width=0.85\linewidth]{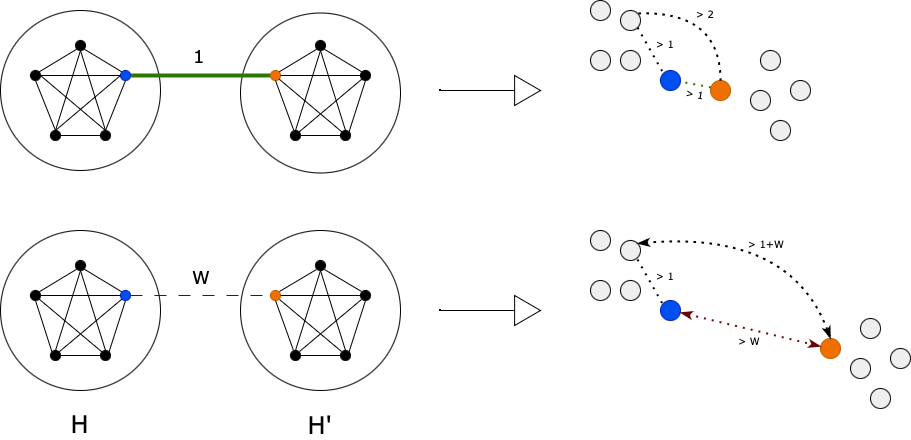}
\centering
\caption{Adversarial sequence of graph updates}
\label{fig:lower_bound}
\end{figure}
Though this sequence is simplistic, it highlights the inherent limitations of embedding into a complete and locally compact metric like the $\ell_p$ normed space. 
We additionally remark that, in the expected distortion setting of our algorithmic result, this lower bound does not persist since a large increase in the expected pairwise distances between nodes does not necessarily imply the embedding for every pair of points has been updated.

\section{Static algorithm}
\label{sec:alg-description}

We proceed to present our algorithm by first presenting the static partitioning
procedure, which is used to initialize our data structure and is subsequently
maintained through the sequence of updates specified by the adversary.
While our ideas are based on prior work, to our knowledge, this static algorithm is a novel construction that has not appeared
before in the literature.

\subsection{Distance Preserving Cuts} \label{sec:dis_pres}

Our algorithm dynamically maintains
an 
embedding based on a set of cuts in the graph,
where each cut is designed to separate vertices with
distances above some threshold $R$,
while simultaneously preserving distances of the vertices in the graph.
We formally define the notion of a \emph{distance preserving cut}.
\begin{definition}[Distance preserving cut] \label{def:dis_pres}
    Given a graph $G=(V, E)$,
    let $S \subseteq V$
    be a random subset of the vertices.
    For vertices $u, v$, let $\cut_{u, v}$
    denote the event that
    $u, v$ are on different
    sides of the partition $(S, V \backslash S)$, i.e.,
    \begin{align*}
        \cut_{u, v} = 
        \cbr{u \in S \text{ and }v \notin S}
        \text{ or }
        \cbr{u \notin S \text{ and }v \in S}
        \ .
    \end{align*}
    We say that $S$
    is a $(\beta, R, \epsilon)$-distance preserving cut, or $(\beta, R, \epsilon)$-cut for short, if
    it has the following three properties:
    \begin{itemize}
      \item $\Pr{\cut_{u, v}} \le \beta \cdot d(u, v)$
        for every $u, v$,
      \item $\Pr{\cut_{u, v}} = 0$
        for every $u, v$ such that $d(u, v) < \epsilon$,
      \item $\Pr{\cut_{u, v}} \ge \frac{1}{2}$
        for every $u, v$ such that $d(u, v) > R$.
    \end{itemize}
\end{definition}

Following in the algorithmic technique of decomposing the graph into these smaller sets with desirable pairwise distance bounds is a refinement on the ball-growing approach of Bartal \cite{bartal1996probabilistic} and the padded decomposition at the heart of Bourgain's embedding results \cite{bourgain1985lipschitz}. Most importantly, this efficient cut set construction allows us to contract edges which are small enough to ignore in the approximation factor and also provide logarithmic distortion bounds on the larger paths -- a fact that will be verified in the following analysis.

The main result in this section is the following lemma that guarantees the existence of such cut sets and will be used heavily in our pre-processing graph decomposition at various granularities which in turn leads to the desired distortion bounds promised in Theorem~\ref{thm:dyn-emb}. %

\begin{lemma}
  For every $1 \le R \le \Delta$, there exists
  a $(\beta, R, \epsilon)$-distance preserving cut with
  $\beta = O\left(\frac{\log n}{R} \right)$ and $\epsilon = \Omega\left(\frac{R}{n}\right)$.
\end{lemma}

We now present the proof of this lemma which uses the
randomized decomposition method of Bartal~\cite{bartal1996probabilistic} in conjunction with a 
novel probabilistic compression procedure. First, we review the definition of an $R$-partition of a graph $G$ (as originally defined in Bartal~\cite{bartal1996probabilistic}).
\begin{definition}
    An $R$-partition of $G$ is a collection of subsets of vertices $P = \{V_1, ..., V_k\}$ such that
    \begin{itemize}
        \item For all $i \in [k]$, $V_i \subseteq V$ and $\bigcup_{i \in [k]} V_i = V$.
        \item For all $i,j \in [k]$ such that $i \neq j$, $V_i \cap V_j = \emptyset$.
        \item Let $C(V_i)$ denote the subgraph of $G$ induced on the vertices
          of $V_i$. Such a subgraph is referred to as a ``cluster'' of the
          partition and for every $i \in [k]$, the weak diameter of $V_i$ is
          bounded above as $\wdiam(C_i) \le R$.
    \end{itemize}
\end{definition}

Next, we give the definition of a $(\beta, R)$-weak decomposition - a modificiation on the $R$ partition that probabilistically ensures vertices close to each other appear in the same cluster.
\begin{definition}
  Given a graph $G=(V, E)$, let $\mC = \cbr{C_1, \dots C_k}$ be a (random) partitioning of the vertices.
  For every $u\in V$, let $C(u) \in [k]$ denote
  the index $i$ such that $u \in C_i$. 
  We say that $\mC$ is a $(\beta, R)$-weak decomposition of $G$ if
  for every $u, v$ we have
  \begin{align*}
      \Pr{C(u) \ne C(v)} \le \beta \cdot d(u, v)
  \end{align*}
  and for every $i$ and pair $u, v \in C_i$ we have
  $d(u, v) \le R$. 
\end{definition}

Following Bartal~\cite{bartal1996probabilistic}, we prove that for all $R$ there exists a $(O(\log n / R), R)$-weak decomposition of $G$ that has the additional property that vertices which are closer than $\frac{R}{2n}$ are necessarily in the same cluster.
Formally, %
\begin{theorem}\label{thm:bartal}
  Given a graph $G = (V,E)$ with parameters $1 \le R \le \Delta$, there exists a $(\beta, R)$-weak decomposition $\{C_1, ..., C_k\}$ with $\beta = O(\log n/ R)$ such that for every pair of vertices $u,v \in V$:
  \begin{itemize}
      \item $\Pr{C(u) \neq C(v)} \le \beta \cdot d_{G}(u,v)$
      \item If $d_{G}(u, v) < \frac{R}{2n}$ then $u$ and $v$ are in the same cluster.
  \end{itemize}
\end{theorem}

\begin{algorithm}[t]
\caption{Low-Diameter Randomized Decomposition (\textsc{Ldrd}) \cite{bartal1996probabilistic}}
\label{alg:bartal}
\begin{algorithmic}[1]
    \STATE \textbf{Input:} Graph $G= (V,E)$ and parameter $1 \le R \le \Delta$
    \STATE \textbf{Output:} An LDRD of G, denoted $\mathcal C = \{C_i\}_{i=1}^n$
    \STATE Contract edges of $G$ which are at most $\frac{R}{2n}$
    \STATE Set $U \leftarrow V$
    \FOR{$u \in U$}
        \STATE Sample $r \sim G(\beta)$
        \STATE $r \gets \min \{r, R\}$
        \STATE $C(u) \leftarrow \{v \in U : d_G(u,v) \le r\}$
        \STATE $U \leftarrow U \setminus C(u)$
    \ENDFOR
    \STATE \textbf{Return:} $\{C_i\}_{i=1}^n$
\end{algorithmic}
\end{algorithm}

Given this randomized decomposition of our graph, we can construct the
desired ``cut'' set that preserves distances by a simple random compression
scheme that combines clusters from the above process. Specifically, we take
each cluster from Theorem~\ref{thm:bartal} and independently assign to one side of the
constructed cut, grouping all the clusters into one of two groups. Within these
groups we then merge the clusters to obtain our desired cut sets, $S$ and $V
\setminus S$. 
The following lemma verifies that this is a distance preserving
cut and the pseudocode is presented in Algorithm~\ref{alg:bre_cut} for clarity. The proof is deferred to the appendix due to space constraints.

\begin{algorithm}[t]
\caption{Randomized $(\beta,R,\epsilon)$-Cut Decomposition}
\label{alg:bre_cut}
\begin{algorithmic}[1]
    \STATE $\mathcal{C} \leftarrow \textsc{Ldrd}(G,R)$
    \STATE $S \leftarrow \emptyset$
    \FOR{$C \in \mathcal C$}
        \STATE Pick $\tau \in \{0,1\}$ uniformly at random
        \IF{$\tau = 1$}
            \STATE $S \leftarrow S \cup C$
        \ENDIF
    \ENDFOR
    \STATE \textbf{Return:} $(S, V\setminus S)$
\end{algorithmic}
\end{algorithm}

\begin{lemma}\label{lem:clu-to-cut}
  Given a value $1 \le R \le \Delta$,
  let $\{C_i\}_{i=1}^{k}$ be the weak decomposition
  of the graph satisfying the properties of Theorem~\ref{thm:bartal}, and 
  define the cut $S$ as
  \begin{math}
    S := \cup_{i \in [k]: x_i=1} C_i,
  \end{math}
  where $x_1, \dots, x_k$ is a sequence of i.i.d Bernoulli variables with parameter $\frac{1}{2}$.
  The cut $S$ is a $(\beta, R, \epsilon)$-distance preserving cut with $\epsilon = R/2n$.
\end{lemma}

\subsection{Embedding Procedure}

We now proceed to show how to obtain an embedding of the graph using the distance preserving cuts of the previous section.
Let $\Delta$ be an upper bound on the diameter of the graph $G$. We define our embedding that builds upon Definition~\ref{def:dis_pres} as follows.
\begin{definition}
    Given a sequence of cuts
    $(S_1, \dots, S_r)$ and
    parameters $(R_1, \dots, R_r)$,
    we define the characteristic embedding
    of $\rbr{S_i, R_i}_{i=1}^{r}$ as a mapping
    $\rho: V \to \R^r$ that sets the $i$-th
    coordinate of $\rho(v)$ to $R_i$ if $v \in S_i$ and to $0$ otherwise, i.e.,
    \begin{math}
        \rho(v)_{i} := R_i \cdot \ind{v \in S_i}.
    \end{math}
\end{definition}
We note the difference our embedding procedure and the existing embedding procedures into $\ell_p$ space. The standard approach to the problem
is to use
\emph{Frechet embeddings};
each coordinate $\rho_i(v)$ is of the form
$d_{G}(v, S_i)$ for some set $S_i$.
These sets are either obtained randomly, or using the partitioning scheme of the Fakcharoenphol-Rao-Talwar (FRT) embedding \cite{frt2004}.
These procedures are not well-suited for the dynamic setting, however because of the analysis of their upper bound on $\norm{\rho(u) - \rho(v)}_p$ for every pair $u, v$. 
Specifically, in order to bound $\norm{\rho(u) - \rho(v)}_p$, the approaches rely on
\begin{align*}
    \abs{\rho_i(u) - \rho_i(v)}
    =
    \abs{d_{G}(u, S_i) - d_{G}(v, S_i)}
    \le d_{G}(u, v)
    \ ,
\end{align*}
where the inequality follows from the triangle inequality. 
In the dynamic setting however, (efficiently) maintaining distances
can only be done approximately. This means that
$\rho_i(u)$ and $\rho_i(v)$ would each be within a
$(1+\epsilon)$ factor of $d_{G}(u, S_i)$ and $d_{G}(v, S_i)$ which would result in a degradation of our guarantees when maintained dynamically.

We now leverage the key characteristics for a set of distance preserving cuts to demonstrate that the corresponding characteristic embedding preserves the original distances in the $\ell_p$-metric with only polylogarithmic distortion.

\begin{theorem}\label{thm:static}
  Given a graph $G=(V, E)$ and a parameter $\Delta$
  that is a power of $2$ satisfying $\Delta \ge \text{diam}(G)$, let
  $S_{1}, \dots, S_{\log(\Delta) +1}$ be (random) subsets of
  $V$ such that
  $S_i$ is a $(\beta_i, R_i, \epsilon_i)$-cut with
  $R_i = 2^{i-2}$ and $(\beta_i, \epsilon_i) = (O(\frac{\log^2 n}{R_i}), \Omega(\frac{R_i}{n}))$,
  and let $\rho: V \to \R^{\log \Delta + 1}$ be the characteristic embedding of these cuts. For every pair of vertices $u, v$ and any $p \in [1, \infty)$:
  \begin{align*}
    \frac14 \cdot d(u, v) \le
    \Ex{\norm{\rho(u) - \rho(v)}_{p}}
    \le
    O(\log^3 n)d(u, v).
  \end{align*}
\end{theorem}
\vspace{-2mm}
We note that the constant $1/4$ can be easily removed by multiplying the embedding vectors by $4$. 
Equipped with this static embedding that only distorts pairwise distances by a polylogarithmic factor, we proceed to adapt the structure to efficiently modify the cut-set decomposition of $G$ through a sequence of (adversarially chosen) edge weight increases.

\section{Dynamic Algorithm}
\label{sec:dynamic}
In this section, we prove Theorem~\ref{thm:dyn-emb}\footnote{Because of the page limit we avoid repeating statements of longer theorems.}. We do it by constructing an algorithm that dynamically maintains an embedding of a metric induced by a graph $G$ into $\ell_{p}$ space of dimension $O(\log{\Delta})$.

Our construction starts by observing a reduction. Informally, in the next theorem we show that in order to maintain dynamically the desired embedding, it is enough to have a dynamic algorithm that for every $1 \le R \le 2\Delta$ maintains a $(O(\frac{\log^2{n}}{R}), R, \Omega(\frac{R}{n}))$-distance preserving cut.
\begin{theorem}\label{thm:dynamic_if_cut}
Assume we are given an algorithm $\mA$
that takes as input the parameter
$R$ and decrementally maintains
a $(\beta, R, \epsilon)$ distance preserving cut $S$
for a graph $G$ undergoing edge weight increases, outputting changes to $S$ after each such update, where $\beta := O(\frac{\log^2{n}}{R})$ and $\epsilon = \Omega(\frac{R}{n})$.
Assume further that
the total running time of the algorithm $\mA$ is bounded by
$t(m, n)$ whp. Then there is a decremental dynamic algorithm that maintains an embedding of
the vertices $\rho: V \to \R^{\log{\Delta} + 1}$, where
$\Delta$ is a power of $2$ always satisfying
$\Delta \ge \text{diam}(G)$\footnote{For instance, we can set $\Delta$ to be the smallest power of $2$ larger than $nW$.},
that
has expected (over the internal randomization of the algorithm) distortion
at most $O(\log^3{n})$ and running time at most $O\left(t(m, n)\log{\Delta} \right)$, whp. 
\end{theorem}

We now show how to maintain a distance preserving cut dynamically, which in turn leads to a dynamic embedding algorithm via Theorem~\ref{thm:dynamic_if_cut} and completes the proof of the main result, Theorem~\ref{thm:dyn-emb}. We start by observing that a $(\beta, \delta)$-weak decomposition of $G$ can be dynamically maintained. We here highlight that the authors of~\cite{forster2021dynamic}, building upon the approach of~\cite{chechik2020dynamic}, have already proved that a $(\beta, \delta)$-weak decomposition of $G$ can be dynamically maintained under \textit{edge deletions} (Corollary 3.8). The proof of our observation involves adapting their techniques to the slightly modified definition of dynamic changes we invoke here to handle the continuous nature of $\ell_p$ space. %

\begin{lemma}\label{lem:dec-decomp}
For every $\beta \in (0, 1)$ and
$\delta = (6(a + 2)(2 + \log m) \ln n)\beta^{-1} = O(a\beta^{-1} \log^2 n)$, where $a \ge 1$ is a given constant controlling the
success probability, there is a decremental algorithm to maintain a probabilistic weak $(\beta, \delta)$-decomposition of a weighted, undirected graph undergoing increases of edge weights that with high probability has total update time $O(m^{1+o(1)} \log W + Q \log n)$, where $Q$ is the total number of updates to the input graph, and (within this running time) is able to report all nodes and incident edges of every cluster that is formed. Over
the course of the algorithm, each change to the partitioning of the nodes into clusters happens by splitting an existing cluster into two or several clusters and each node changes its cluster at most
$O(\log n)$ times.
\end{lemma}

Equipped with this tool, we can present the main contribution of this section - the maintenance of a $(\beta, R, \epsilon)$-distance preserving cut under dynamic edge weights increases.
\begin{lemma} \label{lem:dyn-cut-lem}
For every $0 \le R \le 2\Delta$, there is a decremental dynamic algorithm that maintains a $\left(\beta, R, \epsilon \right)$-distance preserving cut a of weighted, undirected graph $G$ where $\beta = O(\log^2(n)/R)$ and $\epsilon = \Omega(R/n)$. Its total update time is $O(m^{1+o(1)} \log^2 W + Q \log n)$ with high probability, where $Q$ is the total number of updates to the input graph, and, within this running time, explicitly reports all changes to the maintained cut.
\end{lemma} 
The synthesis of these two lemmas with the result of Theorem~\ref{thm:dynamic_if_cut} yields the overall dynamic embedding of Theorem~\ref{thm:dyn-emb}.

\section{Conclusion}
We here present the first dynamic embedding into $\ell_p$ space which is equipped to handle edge weight increases -- a non-trivial extension of the seminal Bourgain and JL embedding results \cite{bourgain1985lipschitz,jl1984}. Most notably, our embeddings produce only a polylogarithmic distortion of the base metric and exhibit an update time on par with the best known results for the APSP and other embedding based problems. Our embedding procedure additionally reports any modifications within polylogarithmic time and is naturally well suited to the class of NP-hard optimization problems which rely on Euclidean geometry for approximations to the optimal solution. To supplement our algorithmic result, we further present a lower bound for the fully dynamic setting where edge weights can be increased or decreased. In particular, we show that no algorithm can achieve a low distortion with high probability without inheriting an update time of $\Omega(n)$ which makes the procedure inefficient in practice.

%

\section*{Impact Statement}
This paper presents work whose goal is to advance the field of Machine Learning. There are many potential societal consequences of our work, none which we feel must be specifically highlighted here.

\section*{Acknowledgements}
We thank the anonymous reviewers for their valuable feedback. 
This work is Partially supported by DARPA QuICC, ONR MURI 2024 award on Algorithms, Learning, and Game Theory, Army-Research Laboratory (ARL) grant W911NF2410052, NSF AF:Small grants 2218678, 2114269, 2347322
Max Springer was supported by the National Science Foundation Graduate Research Fellowship Program under Grant No. DGE 1840340. Any opinions, findings, and conclusions or recommendations expressed in this material are those of the author(s) and do not necessarily reflect the views of the National Science Foundation.

\bibliographystyle{icml2024}
\bibliography{ref}

\begin{thebibliography}{38}
\providecommand{\natexlab}[1]{#1}
\providecommand{\url}[1]{\texttt{#1}}
\expandafter\ifx\csname urlstyle\endcsname\relax
  \providecommand{\doi}[1]{doi: #1}\else
  \providecommand{\doi}{doi: \begingroup \urlstyle{rm}\Url}\fi

\bibitem[Abraham et~al.(2006)Abraham, Bartal, and Neimany]{abraham2006advances}
Abraham, I., Bartal, Y., and Neimany, O.
\newblock Advances in metric embedding theory.
\newblock In \emph{Proceedings of the thirty-eighth annual ACM symposium on
  Theory of computing}, pp.\  271--286, 2006.

\bibitem[Arora et~al.(2005)Arora, Lee, and Naor]{arora2005euclidean}
Arora, S., Lee, J.~R., and Naor, A.
\newblock Euclidean distortion and the sparsest cut.
\newblock In \emph{Proceedings of the thirty-seventh annual ACM symposium on
  Theory of computing}, pp.\  553--562, 2005.

\bibitem[Arora et~al.(2009)Arora, Rao, and Vazirani]{arora2009expander}
Arora, S., Rao, S., and Vazirani, U.
\newblock Expander flows, geometric embeddings and graph partitioning.
\newblock \emph{Journal of the ACM (JACM)}, 56\penalty0 (2):\penalty0 1--37,
  2009.

\bibitem[Aumann \& Rabani(1998)Aumann and Rabani]{aumann1998log}
Aumann, Y. and Rabani, Y.
\newblock An o (log k) approximate min-cut max-flow theorem and approximation
  algorithm.
\newblock \emph{SIAM Journal on Computing}, 27\penalty0 (1):\penalty0 291--301,
  1998.

\bibitem[Awerbuch \& Azar(1997{\natexlab{a}})Awerbuch and Azar]{awerbuch1997}
Awerbuch, B. and Azar, Y.
\newblock Buy-at-bulk network design.
\newblock In \emph{Proceedings 38th Annual Symposium on Foundations of Computer
  Science}, pp.\  542--547, 1997{\natexlab{a}}.
\newblock \doi{10.1109/SFCS.1997.646143}.

\bibitem[Awerbuch \& Azar(1997{\natexlab{b}})Awerbuch and
  Azar]{awerbuch1997buy}
Awerbuch, B. and Azar, Y.
\newblock Buy-at-bulk network design.
\newblock In \emph{Proceedings 38th Annual Symposium on Foundations of Computer
  Science}, pp.\  542--547. IEEE, 1997{\natexlab{b}}.

\bibitem[Bartal(1996)]{bartal1996probabilistic}
Bartal, Y.
\newblock Probabilistic approximation of metric spaces and its algorithmic
  applications.
\newblock In \emph{Proceedings of 37th Conference on Foundations of Computer
  Science}, pp.\  184--193. IEEE, 1996.

\bibitem[Bartal(2004)]{bartal2004}
Bartal, Y.
\newblock Graph decomposition lemmas and their role in metric embedding
  methods.
\newblock In Albers, S. and Radzik, T. (eds.), \emph{Algorithms -- ESA 2004},
  pp.\  89--97, Berlin, Heidelberg, 2004. Springer Berlin Heidelberg.
\newblock ISBN 978-3-540-30140-0.

\bibitem[Bernstein(2009)]{bernstein2009fully}
Bernstein, A.
\newblock Fully dynamic (2+ $\varepsilon$) approximate all-pairs shortest paths
  with fast query and close to linear update time.
\newblock In \emph{2009 50th Annual IEEE Symposium on Foundations of Computer
  Science}, pp.\  693--702. IEEE, 2009.

\bibitem[Bhattacharya et~al.(2022)Bhattacharya, Lattanzi, and
  Parotsidis]{bhattacharya2022efficient}
Bhattacharya, S., Lattanzi, S., and Parotsidis, N.
\newblock Efficient and stable fully dynamic facility location.
\newblock \emph{Advances in neural information processing systems},
  35:\penalty0 23358--23370, 2022.

\bibitem[Blum et~al.(1998)Blum, Konjevod, Ravi, and Vempala]{blum1998semi}
Blum, A., Konjevod, G., Ravi, R., and Vempala, S.
\newblock Semi-definite relaxations for minimum bandwidth and other
  vertex-ordering problems.
\newblock In \emph{Proceedings of the thirtieth annual ACM symposium on Theory
  of computing}, pp.\  100--105, 1998.

\bibitem[Bourgain(1985)]{bourgain1985lipschitz}
Bourgain, J.
\newblock On lipschitz embedding of finite metric spaces in hilbert space.
\newblock \emph{Israel Journal of Mathematics}, 52:\penalty0 46--52, 1985.

\bibitem[Charikar et~al.(1998)Charikar, Chekuri, Goel, and Guha]{charikar1998}
Charikar, M., Chekuri, C., Goel, A., and Guha, S.
\newblock Rounding via trees: Deterministic approximation algorithms for group
  steiner trees and k-median.
\newblock \emph{Conference Proceedings of the Annual ACM Symposium on Theory of
  Computing}, pp.\  114--123, 1998.
\newblock ISSN 0734-9025.
\newblock Proceedings of the 1998 30th Annual ACM Symposium on Theory of
  Computing ; Conference date: 23-05-1998 Through 26-05-1998.

\bibitem[Chawla et~al.(2008)Chawla, Gupta, and R{\"a}cke]{chawla2008embeddings}
Chawla, S., Gupta, A., and R{\"a}cke, H.
\newblock Embeddings of negative-type metrics and an improved approximation to
  generalized sparsest cut.
\newblock \emph{ACM Transactions on Algorithms (TALG)}, 4\penalty0
  (2):\penalty0 1--18, 2008.

\bibitem[Chechik(2018)]{chechik2018near}
Chechik, S.
\newblock Near-optimal approximate decremental all pairs shortest paths.
\newblock In \emph{2018 IEEE 59th Annual Symposium on Foundations of Computer
  Science (FOCS)}, pp.\  170--181. IEEE, 2018.

\bibitem[Chechik \& Zhang(2020)Chechik and Zhang]{chechik2020dynamic}
Chechik, S. and Zhang, T.
\newblock Dynamic low-stretch spanning trees in subpolynomial time.
\newblock In \emph{Proceedings of the Fourteenth Annual ACM-SIAM Symposium on
  Discrete Algorithms}, pp.\  463--475. SIAM, 2020.

\bibitem[Church(2017)]{church2017word2vec}
Church, K.~W.
\newblock Word2vec.
\newblock \emph{Natural Language Engineering}, 23\penalty0 (1):\penalty0
  155--162, 2017.

\bibitem[Dunagan \& Vempala(2001)Dunagan and Vempala]{dunagan2001euclidean}
Dunagan, J. and Vempala, S.
\newblock On euclidean embeddings and bandwidth minimization.
\newblock In \emph{International Workshop on Randomization and Approximation
  Techniques in Computer Science}, pp.\  229--240. Springer, 2001.

\bibitem[D{\"u}tting et~al.(2023)D{\"u}tting, Fusco, Lattanzi, Norouzi-Fard,
  and Zadimoghaddam]{dutting2023fully}
D{\"u}tting, P., Fusco, F., Lattanzi, S., Norouzi-Fard, A., and Zadimoghaddam,
  M.
\newblock Fully dynamic submodular maximization over matroids.
\newblock In \emph{International Conference on Machine Learning}, pp.\
  8821--8835. PMLR, 2023.

\bibitem[Fakcharoenphol et~al.(2003)Fakcharoenphol, Rao, and Talwar]{frt2004}
Fakcharoenphol, J., Rao, S., and Talwar, K.
\newblock A tight bound on approximating arbitrary metrics by tree metrics.
\newblock In \emph{Proceedings of the Thirty-Fifth Annual ACM Symposium on
  Theory of Computing}, STOC '03, pp.\  448–455, New York, NY, USA, 2003.
  Association for Computing Machinery.
\newblock ISBN 1581136749.
\newblock \doi{10.1145/780542.780608}.
\newblock URL \url{https://doi.org/10.1145/780542.780608}.

\bibitem[Feige(1998)]{feige1998approximating}
Feige, U.
\newblock Approximating the bandwidth via volume respecting embeddings.
\newblock In \emph{Proceedings of the thirtieth annual ACM symposium on Theory
  of computing}, pp.\  90--99, 1998.

\bibitem[Forster \& Goranci(2019)Forster and Goranci]{forster2019dynamic}
Forster, S. and Goranci, G.
\newblock Dynamic low-stretch trees via dynamic low-diameter decompositions.
\newblock In \emph{Proceedings of the 51st Annual ACM SIGACT Symposium on
  Theory of Computing}, pp.\  377--388, 2019.

\bibitem[Forster et~al.(2021)Forster, Goranci, and
  Henzinger]{forster2021dynamic}
Forster, S., Goranci, G., and Henzinger, M.
\newblock Dynamic maintenance of low-stretch probabilistic tree embeddings with
  applications.
\newblock In \emph{Proceedings of the 2021 ACM-SIAM Symposium on Discrete
  Algorithms (SODA)}, pp.\  1226--1245. SIAM, 2021.

\bibitem[Forster et~al.(2023)Forster, Goranci, Nazari, and
  Skarlatos]{forster2023bootstrapping}
Forster, S., Goranci, G., Nazari, Y., and Skarlatos, A.
\newblock Bootstrapping dynamic distance oracles.
\newblock \emph{arXiv preprint arXiv:2303.06102}, 2023.

\bibitem[Garg et~al.(2000)Garg, Konjevod, and Ravi]{garg2000}
Garg, N., Konjevod, G., and Ravi, R.
\newblock A polylogarithmic approximation algorithm for the group steiner tree
  problem.
\newblock \emph{Journal of Algorithms}, 37\penalty0 (1):\penalty0 66--84, 2000.
\newblock ISSN 0196-6774.
\newblock \doi{https://doi.org/10.1006/jagm.2000.1096}.
\newblock URL
  \url{https://www.sciencedirect.com/science/article/pii/S0196677400910964}.

\bibitem[Henzinger et~al.(2018)Henzinger, Krinninger, and
  Nanongkai]{henzinger2018decremental}
Henzinger, M., Krinninger, S., and Nanongkai, D.
\newblock Decremental single-source shortest paths on undirected graphs in
  near-linear total update time.
\newblock \emph{Journal of the ACM (JACM)}, 65\penalty0 (6):\penalty0 1--40,
  2018.

\bibitem[Indyk(2001)]{indyk2001algorithmic}
Indyk, P.
\newblock Algorithmic applications of low-distortion embeddings.
\newblock In \emph{Proc. 42nd IEEE Symposium on Foundations of Computer
  Science}, pp.\ ~1, 2001.

\bibitem[J.(2002)]{matousek02}
J., M.
\newblock Lectures on discrete geometry.
\newblock \emph{Graduate Texts in Mathematics}, 2002.
\newblock ISSN 0072-5285.
\newblock \doi{10.1007/978-1-4613-0039-7}.
\newblock URL \url{https://cir.nii.ac.jp/crid/1361981469479209856}.

\bibitem[Johnson et~al.(1986)Johnson, Lindenstrauss, and Schechtman]{jl1984}
Johnson, W.~B., Lindenstrauss, J., and Schechtman, G.
\newblock Extensions of lipschitz maps into banach spaces.
\newblock \emph{Israel Journal of Mathematics}, 54\penalty0 (2):\penalty0
  129--138, 1986.
\newblock \doi{10.1007/BF02764938}.
\newblock URL \url{https://doi.org/10.1007/BF02764938}.

\bibitem[Kleinberg \& Tardos(2002)Kleinberg and Tardos]{kleinberg2002}
Kleinberg, J. and Tardos, E.
\newblock Approximation algorithms for classification problems with pairwise
  relationships: Metric labeling and markov random fields.
\newblock \emph{J. ACM}, 49\penalty0 (5):\penalty0 616–639, sep 2002.
\newblock ISSN 0004-5411.
\newblock \doi{10.1145/585265.585268}.
\newblock URL \url{https://doi.org/10.1145/585265.585268}.

\bibitem[Krauthgamer et~al.(2004)Krauthgamer, Lee, Mendel, and
  Naor]{krauthgamer2004measured}
Krauthgamer, R., Lee, J.~R., Mendel, M., and Naor, A.
\newblock Measured descent: A new embedding method for finite metrics.
\newblock In \emph{45th Annual IEEE Symposium on Foundations of Computer
  Science}, pp.\  434--443. IEEE, 2004.

\bibitem[Larsen \& Nelson(2017)Larsen and Nelson]{larsen2017optimality}
Larsen, K.~G. and Nelson, J.
\newblock Optimality of the johnson-lindenstrauss lemma.
\newblock In \emph{2017 IEEE 58th Annual Symposium on Foundations of Computer
  Science (FOCS)}, pp.\  633--638. IEEE, 2017.

\bibitem[Leskovec \& Krevl(2014)Leskovec and Krevl]{snapnets}
Leskovec, J. and Krevl, A.
\newblock {SNAP Datasets}: {Stanford} large network dataset collection.
\newblock \url{http://snap.stanford.edu/data}, June 2014.

\bibitem[Linial et~al.(1995)Linial, London, and Rabinovich]{linial1995geometry}
Linial, N., London, E., and Rabinovich, Y.
\newblock The geometry of graphs and some of its algorithmic applications.
\newblock \emph{Combinatorica}, 15:\penalty0 215--245, 1995.

\bibitem[Roditty \& Zwick(2004)Roditty and Zwick]{roditty2004dynamic}
Roditty, L. and Zwick, U.
\newblock On dynamic shortest paths problems.
\newblock In \emph{Algorithms--ESA 2004: 12th Annual European Symposium,
  Bergen, Norway, September 14-17, 2004. Proceedings 12}, pp.\  580--591.
  Springer, 2004.

\bibitem[Roditty \& Zwick(2012)Roditty and Zwick]{roditty2012dynamic}
Roditty, L. and Zwick, U.
\newblock Dynamic approximate all-pairs shortest paths in undirected graphs.
\newblock \emph{SIAM Journal on Computing}, 41\penalty0 (3):\penalty0 670--683,
  2012.

\bibitem[Subercaze et~al.(2015)Subercaze, Gravier, and
  Laforest]{subercaze2015metric}
Subercaze, J., Gravier, C., and Laforest, F.
\newblock On metric embedding for boosting semantic similarity computations.
\newblock In \emph{Proceedings of the 53rd Annual Meeting of the Association
  for Computational Linguistics and the 7th International Joint Conference on
  Natural Language Processing (Volume 2: Short Papers)}, pp.\  8--14, 2015.

\bibitem[Thorup \& Zwick(2005)Thorup and Zwick]{thorup2005approximate}
Thorup, M. and Zwick, U.
\newblock Approximate distance oracles.
\newblock \emph{Journal of the ACM (JACM)}, 52\penalty0 (1):\penalty0 1--24,
  2005.

\end{thebibliography}

\newpage
\onecolumn
\appendix

\section{Empirical validation}
We tested the theoretical algorithm guarantees on three different graphs.

\textit{Data sets preparation.} As the backbone for each graph, we used the social network of LastFM users from Asia available in the Stanford Network Analysis Project dataset (SNAP)~\cite{snapnets}. To adhere to our dynamic setting, we randomly chose a subset of 150, 300, and 600 connected nodes to form three different bases of the dynamically changing network. We added random weights from a uniform distribution to these graphs. We augmented each graph by respectively 10000, 5000, and 1000 changes to the topology (queries). Each change increases the weight of a randomly and uniformly chosen edge of the graph by a number chosen from a uniform distribution whose range increases as the process progresses. 

\textit{Evaluation.} We implemented the cut-preserving embedding from Theorem~\ref{thm:dyn-emb} and computed the distances between every pair of nodes in the graph after each query. We compared the average of these distances with the average distances computed by an exact algorithm that in an offline fashion computes the shortest distances after each query. Visualized results are presenting in Figure~\ref{fig:figure3}.
\begin{figure}[htp]
\centering
\includegraphics[width=.32\textwidth]{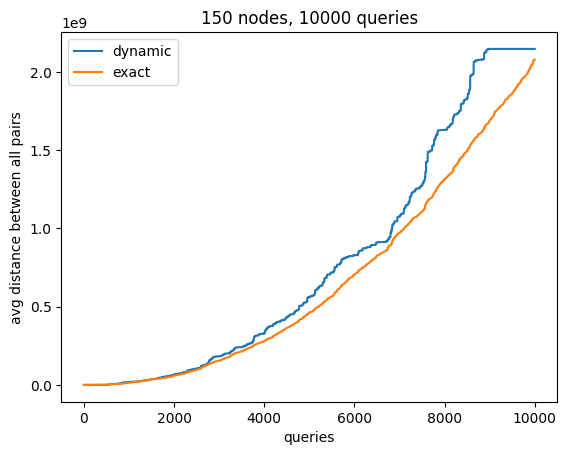}\hfill
\includegraphics[width=.32\textwidth]{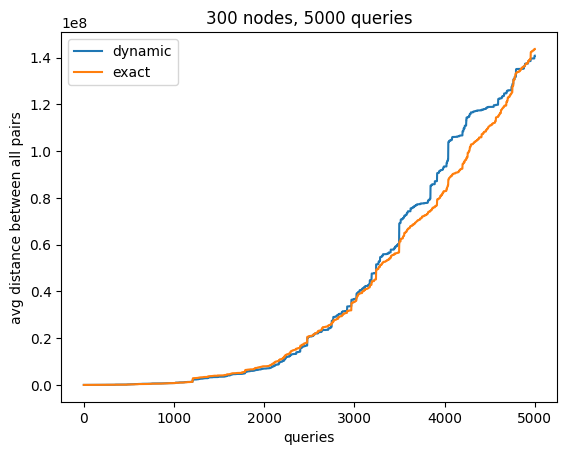}\hfill
\includegraphics[width=.32\textwidth]{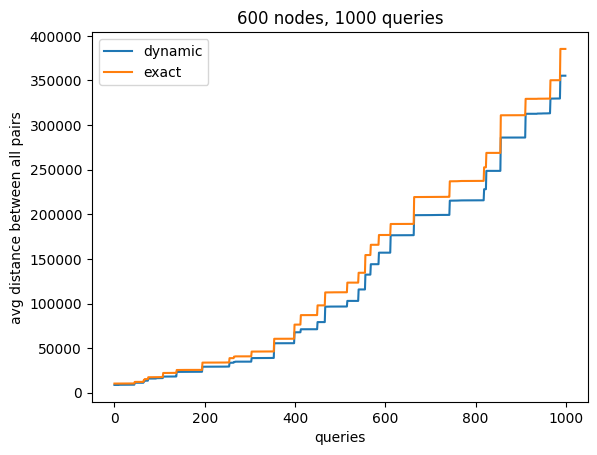}

\caption{Visualization of average distances in a dynamically changing metric. The orange line represents the average distance between all pairs of nodes computed exactly, using a deterministic shortest path algorithm, after every query. The blue line represents the average distance computed based on the embedding given by the dynamic embedding algorithm proposed in the paper.}
\label{fig:figure3}
\end{figure}
To allow more direct reasoning about the distortion of our embedding, in Figure~\ref{fig:figure-ratio} we provide plots representing, after each query, the ratio of the average distance based on our dynamic embedding to the average distance computed exactly. 
We would like to note that even though theoretically our embedding is not contractive, this property holds in expectation. In practice, small fluctuations may appear which are particularly visible in the case of a small number of queries.
\begin{figure}[htp]
\centering
\includegraphics[width=.32\textwidth]{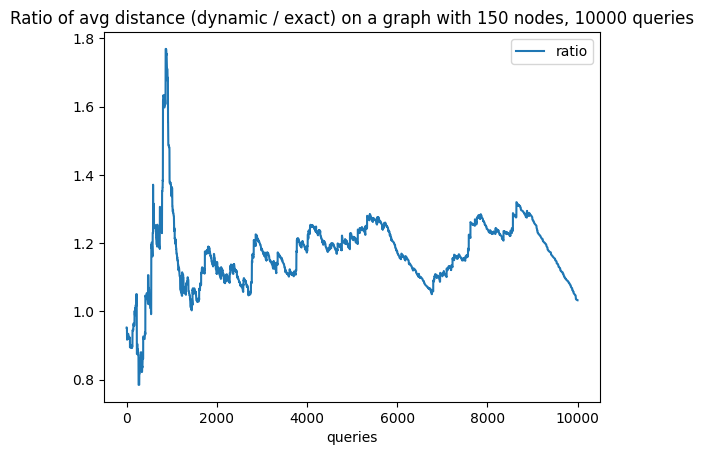}\hfill
\includegraphics[width=.32\textwidth]{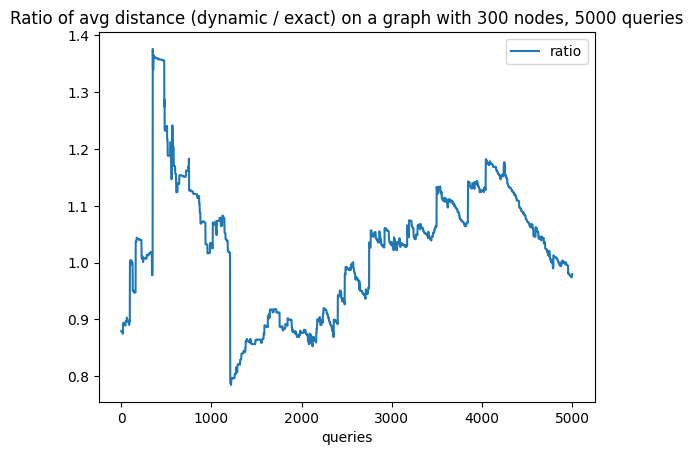}\hfill
\includegraphics[width=.32\textwidth]{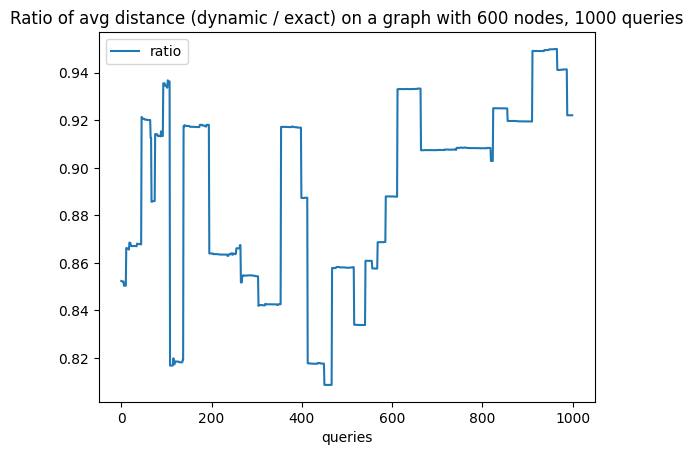}
\caption{The ratio of the average distance between all pairs of points computed by of our embedding to the exact average distance between all pairs, after each query.}
\label{fig:figure-ratio}
\end{figure}

\textit{Conclusions.} We observed that in these experiments the achieved stretch is within a constant factor of the average over the exact distances which adheres to the $O(log^2(n))$ theoretical bound of Theorem~\ref{thm:dyn-emb} and even surpasses it. This advantage might be a consequence of a couple of things, e.g. the random process involved in the generation, or a few number of testing instances. Nevertheless, we find this results promising. We hope that they can serve as a starting point for further investigation of practicality of dynamic embedding.

\section{Omitted Proofs}
\subsection{Static Algorithm}
We briefly provide a sketch of for the proof of Theorem~\ref{thm:bartal} below and refer to Bartal~\cite{bartal1996probabilistic} for the full details. 

\begin{proof}[Proof sketch]
We construct the so-called ``low diameter randomized decomposition'' (LDRD) for the input graph $G$ with (parameter $R$) via the following high-level procedure: first, contract all paths between vertices that are shorter than $\frac{R}{2n}$. Now, starting with the contracted graph we pick an arbitrary vertex, $u$, and select a radius $r$ sampled from the geometric distribution with success parameter $\beta$; if $r \ge R$ then we simply replace it with $R$. Note that, with high probability, $r$ will not be replaced because of the choice of the parameters.
Mark all unmarked vertices contained within the set $B_r(u) = \{v \in G : d(u,v) \le r\}$ and define this to be a new cluster of the partition. Repeat this ball-growing process with the next unmarked vertex in $G$ and only consider the remaining unmarked vertices to be included in future clusters. This procedure is repeated until all vertices are marked and return the created clusters. Pseudocode for this is provided in Algorithm~\ref{alg:bartal}. 

To verify the theorem, we proceed to prove each property of a $(\beta,R)$-weak decomposition holds with the additional probabilistic bound. The second property holds by contraction of small edges. Therefore, we need only obtain an upper bound on the probability that the two vertices $u$ and $v$ are not assigned to the same cluster, ie., $C(u) \neq C(v)$. This final point is a standard procedure attributed to \cite{bartal1996probabilistic} which we overview below.

For some edge $e = (u,v)$, we need to bound the probability that $e$ is contained within none of the clusters of the decomposition. Specifically, we need to ensure that after each ball growing procedure (Line 8 of Algorithm~\ref{alg:bartal}) $e$ is not added to a cluster. We can therefore decompose this value as the probability that exactly one end point of $e$ was contained in the last ball grown or neither is contained. Let $t$ denote the stage of the ball growing procedure and let $X_t$ denote the event that edge $(u,v) \notin C_j$ for all $j \le t$. Then we can recursively compute $\Pr{X_t}$ as the probability that exactly one of $u,v$ are contained in $C_t$, or the probability that neither is as a condition on $\Pr{X_{t+1}}$. The direct computation of these values using the probability density functions for the radius random variable's distribution and bounding the probability that an endpoint of $(u,v) \in C_t$ is contained in Section 3 of \cite{bartal1996probabilistic} and yields the desired $O\left( \frac{\log n}{R}\right)$ upper bound on the probability that $C(u) \neq C(v)$.

\end{proof}

\begin{proof}[Proof of Lemma~\ref{lem:clu-to-cut}]
  By the guarantee on the partition given by Theorem~\ref{thm:bartal}, for every pair of vertices $u,v \in V$ we must have that 
  $\Pr{u \text{ and } v \text{ in different clusters}} \le 
  \beta \cdot d(u,v)$,
  and 
  $\Pr{u \text{ and } v \text{ in different clusters}} =0$
  for all $u, v$ such that $d_{G}(u, v) < \frac{R}{n}$. 
  Note that if the two nodes are in the same cluster than they must
  be on the same side of the cut. 
  Therefore, the first two properties of $(\beta, R)$-cuts are proved.

  Now, further assume that $u,v \in V$ such that $d(u,v) > R$. By construction,
  for all
  $C_i$, have $\wdiam(C_i) \le R$. Therefore, the vertices $u$ and $v$
  cannot be contained in the same cluster. Let $C_u$ and $C_v \ne C_u$ denote
  the clusters containing $u$ and $v$ respectively. 
  By construction of the cut set $S$, we have that $\Pr{\cut_{C_u, C_v}} =
  \frac12$. Therefore, $\Pr{\cut_{u, v}} \ge 1/2$ as claimed.
\end{proof}

\begin{proof}[Proof of Theorem~\ref{thm:static}]
We divide the proof into a few key Claims.
For every pair of vertices $u, v$, define
$d'_p(u, v) := \norm{\rho(u) - \rho(v)}_{p}$. 
We first lower bound $\Ex{d'_p(u, v)}$.

\begin{claim}
For every pair of vertices $u, v\in V$:
\begin{align*}
  \Ex{d'_p(u, v)} \ge \Ex{d'_{\infty}(u, v)} \ge \frac{d(u, v)}{4}
  \ .
\end{align*}
\end{claim}

\begin{proof}    
    The first inequality holds for any embedding $\rho$ since
    $\norm{.}_{p}$ is decreasing in $p$; we therefore focus on the second inequality. 
    We assume that $d(u, v) \ne 0$ as otherwise the claim is trivially true.
    
    Fix $u, v$ and  
    set $i$ to be the maximal integer satisfying
    $2^{i-2} < d(u, v)$. 
    By definition of $i$, we must have
    $2^{i - 1} \ge d(u, v)$.
    We note that
    $i \in [\log \Delta + 1]$;
    we have $i\ge 1$ because
    $2^{1 - 2} < 1 \le d(u, v)$ and
    we have $i < \log \Delta + 2$ because
    $2^{(\log \Delta + 2) -2} =\Delta \ge d(u, v)$.

    Consider the cut $S_i$. 
    Since $d(u, v) > R_i$,
    by definition of a distance
    preserving cut, we have 
    $\Pr{\cut_{u, v}(S_i)} \ge \frac{1}{2}$ which implies
    that with probability at least $\frac{1}{2}$,
    the $i$-th coordinate 
    $\rho(u)$ and $\rho(v)$
    different. This implies
    the claim by the definition of the
    characteristic embedding.
    Formally,
    \begin{align*}
        \Ex{d'_{\infty}(u, v)} &\ge
        \Ex{\abs{\rho_{i}(u) - \rho_i(v)}}
        &\EqComment{Definition of $\norm{.}_{\infty}$}
        \\&=
        R_i \cdot \Pr{\cut_{u, v}(S_i)}
        &\EqComment{Definition of $\rho_i$}
        \\&\ge
        \frac{R_i}{2}
        &\EqComment{Assumption on $S_i$}
        \\&\ge
        \frac{d(u, v)}{4}
        &\EqComment{Since $2R_i \ge d(u, v)$}
        .
    \end{align*}
    thus, proving the claim.
\end{proof}
We now upper bound $\Ex{d'_p(u, v)}$.
\begin{claim}
For every pair of vertices $u,v \in V$:
    \begin{align*}
        \Ex{d'_p(u, v)} \le
        \Ex{d'_1(u, v)} \le O(\log^2n) \cdot d(u, v)
        \ .
    \end{align*}
\end{claim}
\begin{proof}
    As before, first inequality follows from
    the fact that $\norm{.}_{p}$ is decreasing in $p$. We therefore focus on the second inequality.
    For convenience, we denote $d := d(u,v)$. Note that, for every scale $i$, we have
    \begin{align}
        \Ex{\norm{\rho_{i}(u) - \rho_{i}(v)}_1}
        &= \Ex{R_i \cdot \ind{\cut_{u, v}(S_i)}} \notag \\
        &= R_i \cdot \Pr{\cut_{u, v}(S_i)}. \label{eq:prob_clus_bound}
    \end{align}
    By definition, the $\ell_1$ norm will merely be a summation on the above over the scales of $R_i$:
    \begin{align*}
        \Ex{\norm{\rho(u) - \rho(v)}_1} &= \sum_{i=1}^{\log \Delta + 1}\Ex{| \rho_i(u)) - \rho_i(v)|} \\
        &= \sum_{i = 1}^{\log \Delta + 1} R_i \cdot \Pr{\cut_{u, v}(S_i)}.
    \end{align*}
    We proceed to bound this summation by bounding individual summands corresponding to manageable scales with the properties of our distance preserving cuts from Section~\ref{sec:dis_pres}.

    We divide the above sum into three parts depending on the relationship the parameters and $d$:
    \begin{itemize}
        \item Case 1: $d > R_i$.
        For these $i$,
        we simply bound
        $\Pr{\cut_{u, v}(S_i)}$ with $1$.
        \item Case 2: $d < \epsilon_i$.
        In this case,
        since the cut $S_i$ is distance preserving, we must have $\Pr{\cut_{u, v}(S_i)}=0$.
        \item Case 3: $\epsilon_i \le d \le R_i$.
        Since $S_i$ is distance preserving, we
        we know that
        $\Pr{\cut_{u, v}(S_i)}\le \beta_i d(u, v)$.
    \end{itemize}

    We proceed to combine the three cases. We first note that any $i$ in Case $1$ satisfies $R_i < d$
    which means
    $2^{i-2} < d$
    or equivalently
    $i < \log(d) + 2$.
    Therefore,
    \begin{align*}
        \sum_{i: R_i < d}
        R_i
        \le 
        \sum_{i: 2^{i-2} < d}
        2^{i-2}
        \le 
        \sum_{i=1}^{\ceil{\log d} + 2}
        2^{i-2}
        \le 2^{\ceil{\log d} +1}
        = O(d).
    \end{align*}

    As for Case 3,  
    \begin{align}
        \sum_{i: \epsilon_i \le d \le R_i}
        R_i \Pr{\cut_{u, v}(S_i)}
        \le 
        \sum_{i: \epsilon_i \le d \le R_i}
        \beta_i R_i d(u, v)
        = 
        d(u, v) O(\log^2 n)
        \cdot \abs{i: \epsilon_i \le d \le R_i},
        \label{eq:jul15_1548}
    \end{align}
    where fore the final inequality we have used the assumption $\beta_i \le O( \frac{\log^2 n}{R_i})$.

    We proceed to bound 
    $\abs{i: \epsilon_i \le d \le R_i}$.
    Let $c$ be the constant
    such that $\epsilon_i \ge c R_i/n$.
    In order for $d$ to satisfy
    $d\in [\epsilon_i, R_i]$, we must have
    $d \le 2^{i-2}$ which means
    $i \ge \log(d) + 2$ and
    $c 2^{i-2}/n \le d$ which means
    $i \le \log(nd/c) + 2$.
    It follows that
    \begin{align*}
        \abs{i: \epsilon_i \le d \le R_i}
        \le 
        (\log(nd/c) + 2)
         - (\log(d) + 2) + 1
         = \log(n) + \log(c^{-1}) + 1,
    \end{align*}
    which is $O(\log(n) + 1) = O(\log n)$. 
    Plugging this back in Equation~\eqref{eq:jul15_1548}
    we obtain
    \begin{align*}
        \sum_{i: \epsilon_i \le d \le R_i}
        R_i \Pr{\cut_{u, v}(S_i)}
        \le O(\log^3 n) d(u, v)
    \end{align*}

    Combining the above cases,
    and using Equation~\eqref{eq:prob_clus_bound}, we conclude that
    \begin{align*}
        \sum_{i=1}^{\log \Delta + 1}
        \abs{\rho_i(u) - \rho_i(v) }
        \le O(\log^3(n)) d(u, v) .
    \end{align*}
\end{proof}
Combining the two claims, we have the result of Theorem~\ref{thm:static}.
\end{proof}

\subsection{Dynamic Algorithm}
\begin{proof}[Proof of Theorem~\ref{thm:dynamic_if_cut}]
The decremental algorithm uses $\log(\Delta) + 1$ instances of algorithm $\mathcal{A}$ with the given choice of parameters which allows us to follow the argument provided for the static case in Section~\ref{sec:alg-description}. Let $R_{i} = 2^{i}$ for $1 \le i \le \log(\Delta)$. The $i$-th instance of the algorithm $\mathcal{A}$ maintains a $\left(O(\frac{\log^2{n}}{R_{i}}), R_{i}, \Omega(\frac{R_{i}}{n}) \right)$-distance preserving cut $S_{i}$. The final embedding $\rho : V \rightarrow \R^{\log{\Delta}}$ is the characteristic embedding of the cuts $(S_{1}, \ldots, S_{\log(\Delta) + 1})$ and parameters $\left(R_{1}, \ldots, R_{\log(\Delta) + 1} \right)$. 
Formally, we set the $i$-th coordinate of $\rho(v)$ to be $R_i$ if $v \in S_i$ and to $0$ otherwise.

Upon the arrival of a decremental change, the algorithm inputs this change into every instance of the dynamic decremental algorithm $\mathcal{A}$ it runs.
By assumption on the input algorithm $\mathcal{A}$, each run explicitly outputs changes to the structure of the $i$-th cut. Therefore, the main algorithm can adapt to these changes by appropriately changing the coordinates of vertices. Specifically,  if a vertex is removed from the $i$-th cut its coordinate changes from $0$ to $R_{i}$ and if it is added the opposite change takes place.
Since processing each such change takes constant time, and there are $t(m, n)$ updates in total by assumption on $\mathcal{A}$, the total time for the $i$-th instance is $t(m, n)$.
Therefore,
by charging the time used for these changes to the total update time of $O(\log{\Delta})$ instances of the algorithm $\mathcal{A}$, the total running time of the decremental algorithm is at most $O\left(t(m, n) \log{\Delta}\right)$.
As for the distortion, by assumption on algorithm $\mathcal{A}$, each run maintains a $\left(O(\frac{\log^2{n}}{R_{i}}), R_{i}, \Omega(\frac{R_{i}}{n})\right)$-distance preserving cut. Thus, the stretch of the maintained embedding $\rho$ follows from applying Theorem~\ref{thm:static}.
\end{proof}

\begin{proof}[Proof sketch of Lemma~\ref{lem:dec-decomp}]
At a high level, the dynamic algorithm presented in~\cite{forster2021dynamic} for maintaining a weak decomposition undergoing edge deletions relies on the concept of assigning a center to each cluster in the decomposition, an idea initially introduced in~\cite{chechik2020dynamic}.

This technique employs a dynamic Single Source Shortest Paths (SSSP) algorithm (specifically, the SSSP algorithm described in~\cite{henzinger2018decremental}) to monitor the distances from a center to every vertex within the cluster. 
Whenever an edge is deleted, the change is also updated to the SSSP algorithm. 
The SSSP algorithm then outputs vertices from the cluster whose distance to the center is greater than a certain threshold, and such that keeping these vertices within the cluster could potentially violate the requirements of a $(\beta, \delta)$-weak decomposition. 
To prevent this event, an appearance of such a vertex incurs either a re-centering of the current cluster or splitting the cluster into two or more disjoint new clusters. These operations ensure that eventually the diameter of each cluster satisfies the requirements of the $(\beta, R)$-weak decomposition. Crucially, the authors show that the number of times a cluster is re-centered can be bounded by $a\log{n}$, for some absolute constant $a$. On the other hand, the splitting procedure is designed in such a way that the size of a cluster that is split from the previous cluster shrinks by at least a factor of $2$. As a result, any vertex can be moved to a new cluster at most $O(\log{n})$ times. 

Now, the crucial observation that allows us to carry the approach to the case of edge weight increases is the fact that the SSSP algorithm of~\cite{henzinger2018decremental} also handles edge weight increases while preserving the same complexity bounds. The algorithm then is the same as in~\cite{forster2021dynamic} with the only change that, to monitor distance from a center of a cluster to every other vertex, we use the SSSP version that supports edge weights increases. As for the correctness part, we can carry the analysis from~\cite{forster2021dynamic} with minor changes. In Lemma~3.3, we show that the probability of being an inter-cluster edge is at most $\beta w_{k}(e)$, where $w_{k}(e)$ denotes the weight of edge $e$ in the current dynamic graph $G_{k}$ after $k$ updates. This, by the union bound, implies that for every pair of vertices $u,v \in G_{k}$ it holds $\Pr{C(u) \neq C(v)} \le \beta \cdot d_{G_{k}}(u, v)$. From the fact that weight updates only increase distances, we observe that Lemmas 3.5 and 3.6 from~\cite{forster2021dynamic} still hold (here, it is also important that updates are independent of the algorithm, which is also true in our model). This observation ultimately gives us that each cluster undergoes a center re-assignment at most $O(a \cdot \log{n})$ times. As a consequence, our derivation of total update time follows from the one in~\cite{forster2021dynamic} with the change that we account $O(Q \log n)$ time to parse all updates since each edge is present in at most $O(\log n)$ different SSSP instances in their algorithm.
\end{proof}

\begin{proof}[Proof of Lemma~\ref{lem:dyn-cut-lem}]
Let the graph $G'$
denote the graph obtained
by replacing all edges in $G$
with weight $\le \epsilon$ with
weight $0$.
It is easy to see that we can
maintain the graph $G$ with $\Theta(m + Q)$ total update time;
we start by preprocessing $G$ at the beginning of the algorithm. Next, whenever there is an update to an edge weight, we pass the update along to $G'$ if the new values of the weight is strictly larger than $\epsilon$ and ignore it if it is at most $\epsilon$. Note that in the latter case the old value is at most $\epsilon$ as well since we only allow edge-weight increases.

We then
use the  dynamic decremental algorithm from Lemma~\ref{lem:dec-decomp} for the choice of parameters $(\beta, R/2)$ and $a = \Theta(1)$ on the graph $G'$.
As a consequence of the initialization, a $\left(\beta, R/2
\right)$-weak decomposition of the preprocessed $G'$ is computed.
Denote the clusters of this decomposition $C_{1}, \ldots, C_{k}$. 
We claim that
the clusters satisfy the following properties.
\begin{itemize}
    \item The weak diamater of any $C_i$ in $G$ is at most $R$.
    \item For any $u, v$,
    the probability that they are in different clusters is at most $\beta d_{G}(u, v)$.
    \item Any $u, v$ such that $d_{G}(u, v) \le \epsilon$ are in the same cluster.
\end{itemize}
We start with the first property.
By definition of a $(\beta, R/2)$-weak decomposition,
if $d_{G'}(u, v) > R/2$ then $u$ and $v$ are in different clusters. We note however that
$d_{G}(u, v) \le d_{G'}(u, v) + \epsilon n$; this is because any path connecting $u$ and $v$ in $G'$ has at most $n$ edges that with replaced weights. Choosing $\epsilon \le R/2n$, it follows that $d_{G}(u, v) \le R$ for all $u$ and $v$ in the same cluster.
For the second property, we note that
the probability that $u$ and $v$ are in different clusters is at most $\beta d_{G'}(u, v)$, and therefore the property follows from $d_{G'}(u, v) \le d_{G}(u, v)$.
Finally, for the third property, it holds because any two such vertices will have distance $0$ in $G'$ and therefore are always in the same cluster.

After initialization, the algorithm samples $k$ uniform and independent values from $\{0,1\}$, one value for each cluster. 
Next, a cut of $G$ is created by grouping vertices from clusters that have been assigned value $1$, denoting this side of the cut $S$. By Lemma~\ref{lem:clu-to-cut}, we have that the cut $S$ is a $\left(\beta, R, \epsilon \right)$-distance preserving cut. We also note that the cut $S$ can be generated with $O(k) = O(n)$ additive overhead to the time complexity of the dynamic decremental algorithm from Lemma~\ref{lem:dec-decomp}.

Finally, we discuss the algorithm action upon a decremental update of an edge.
As mentioned, we maintain $G'$ by forwarding the update of an edge if necessary depending on whether the new weight is $\le \epsilon$ or $> \epsilon$.
The update to $G'$ is then in turn forwarded to the algorithm that dynamically maintains $\left(\beta, R/2\right)$-weak decomposition of $G'$. 
According to Lemma~\ref{lem:dec-decomp}, the only changes to the partitioning occur when splitting an existing cluster into two or more new clusters and members of the new clusters need to be explicitly listed. Let $\mathcal{C}' = \{C'_{1}, \ldots, C'_{j}\}$ be the set of these newly formed clusters. The main algorithm temporarily deletes the vertices belonging to clusters $\mathcal{C}'$ from the dynamic cut $S$ it maintains. Then, for each new cluster, it samples independently and uniformly a value from $\{0,1\}$. Again, vertices from clusters that have been sampled $1$ are added to those who already had $1$ (before the decremental update), while the ones who have just sampled $0$ are assigned to the other side of the cut. 
Since the decremental change is oblivious to the randomness of the algorithm, the random bits assigned to all clusters of the new partitioning obtained from the decremental update are distributed i.i.d. according to a Bernoulli distribution with parameter $\frac{1}{2}$, and by applying Lemma~\ref{lem:clu-to-cut} we obtain that the updated cut is also a $(\beta, R, \epsilon)$-distance preserving cut. As for the time complexity, we can observe that the number of changes made to the structure is linearly proportional to the number of vertices changing a cluster after an update. It, therefore, follows that the total update time needed for maintaining the cut is dominated by the time needed to report the changes in the dynamic partitioning, and according to Lemma~\ref{lem:dec-decomp}, is at most $O\left(m^{1 + o(1)}\log^2 W + Q\log n\right)$.
\end{proof}

\end{document}